\documentclass[a4paper,UKenglish,cleveref, autoref]{lipics-v2019}

\usepackage{bm}
\usepackage{xspace}
\usepackage[textwidth=2cm,textsize=small]{todonotes}
\usepackage{tikz}
\usetikzlibrary{arrows,shapes,automata,
calc,chains,matrix,positioning,scopes,fit,backgrounds}

\renewcommand{\leq}{\leqslant}
\renewcommand{\geq}{\geqslant}
\renewcommand{\le}{\leqslant}
\renewcommand{\ge}{\geqslant}

\newcommand*{\eg}{e.g.\@\xspace}
\newcommand*{\ie}{i.e.\@\xspace}

\def\cqedsymbol{\ifmmode$\lrcorner$\else{\unskip\nobreak\hfil
\penalty50\hskip1em\null\nobreak\hfil$\lrcorner$
\parfillskip=0pt\finalhyphendemerits=0\endgraf}\fi}
\newcommand{\cqed}{\renewcommand{\qed}{\cqedsymbol}}

\title{Reachability for Bounded Branching VASS} 

\author{Filip Mazowiecki}{LaBRI, Université de Bordeaux, France}{filip.mazowiecki@u-bordeaux.fr}{}{This study has been carried out with financial support from the French State, managed by the French National Research Agency (ANR) in the frame of the ``Investments for the future'' Programme IdEx Bordeaux  (ANR-10-IDEX-03-02).}

\author{Micha\l{} Pilipczuk}{University of Warsaw, Poland}{michal.pilipczuk@mimuw.edu.pl}{}{This work is 
a part of project TOTAL that has received funding from the European Research Council (ERC) 
under the European Union's Horizon 2020 research and innovation programme, grant agreement No.~677651.}

\authorrunning{J.-F. Mazowiecki and Mi. Pilipczuk}

\Copyright{Filip Mazowiecki and Micha\l{} Pilipczuk}



\category{}

\relatedversion{An arXiv version of the paper is available at \url{https://arxiv.org/abs/1904.10226}.}

\supplement{}

\acknowledgements{}

\nolinenumbers 

\EventEditors{Wan Fokkink and Rob van Glabbeek}
\EventNoEds{2}
\EventLongTitle{30th International Conference on Concurrency Theory (CONCUR 2019)}
\EventShortTitle{CONCUR 2019}
\EventAcronym{CONCUR}
\EventYear{2019}
\EventDate{August 27--30, 2019}
\EventLocation{Amsterdam, the Netherlands}
\EventLogo{}
\SeriesVolume{140}
\ArticleNo{24}

\newcommand{\bovass}{$\textsc{BoVASS}$\xspace}
\newcommand{\onebovass}{$1-\textsc{BoVASS}$\xspace}
\newcommand{\twobovass}{$2-\textsc{BoVASS}$\xspace}
\newcommand{\vass}{$\textsc{VASS}$\xspace}
\newcommand{\onevass}{$1-\textsc{VASS}$\xspace}
\newcommand{\twovass}{$2-\textsc{VASS}$\xspace}
\newcommand{\doublebovass}{$1-\textsc{BoVASS}$^{\times 2}\xspace}
\newcommand{\divbovass}{$1-\textsc{BoVASS}$^{\div 2}\xspace}
\newcommand{\doubledivbovass}{$1-\textsc{BoVASS}$^{\times 2,\div 2}\xspace}

\newcommand{\bobr}{$\textsc{BoBrVASS}$\xspace}
\newcommand{\onebobr}{$1-\textsc{BoBrVASS}$\xspace}
\newcommand{\twobobr}{$2-\textsc{BoBrVASS}$\xspace}
\newcommand{\bvass}{$\textsc{BrVASS}$\xspace}
\newcommand{\onebvass}{$1-\textsc{BrVASS}$\xspace}
\newcommand{\twobvass}{$2-\textsc{BrVASS}$\xspace}
\newcommand{\doublebobr}{$1-\textsc{BoBrVASS}$^{\times 2}\xspace}
\newcommand{\divbobr}{$1-\textsc{BoBrVASS}$^{\div 2}\xspace}
\newcommand{\doubledivbobr}{$1-\textsc{BoBrVASS}$^{\times 2,\div 2}\xspace}

\newcommand{\V}{\mathcal{V}}
\newcommand{\B}{\mathcal{B}}

\newcommand{\set}[1]{\{ #1 \}}
\newcommand{\N}{\mathbb{N}}
\newcommand{\Z}{\mathbb{Z}}

\newcommand{\nn}{\bm{n}}
\newcommand{\mm}{\bm{m}}
\newcommand{\zz}{\bm{z}}
\newcommand{\vv}{\bm{v}}
\newcommand{\ww}{\bm{w}}
\newcommand{\vzero}{\bm{0}}

\begin{document}

\maketitle


\begin{abstract}
In this paper we consider the reachability problem for bounded branching VASS. Bounded VASS are a variant of the classic VASS model where all values in all configurations are upper bounded by a fixed natural number, encoded in binary in the input. This model gained a lot of attention in~2012 when Haase et al. showed its connections with timed automata. Later in~2013 Fearnley and Jurdzi\'{n}ski proved that the reachability problem in this model is PSPACE-complete even in dimension~$1$. Here, we investigate the complexity of the reachability problem when the model is extended with branching transitions, and we prove that the problem is EXPTIME-complete when the dimension is~$2$ or larger.
\end{abstract}

\vskip -0.7cm
\begin{picture}(0,0)
\put(392,25)
{\hbox{\includegraphics[width=40px]{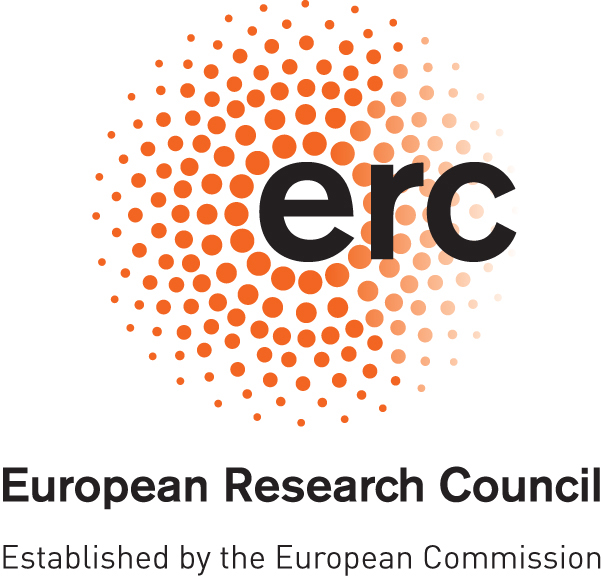}}}
\put(382,-35)
{\hbox{\includegraphics[width=60px]{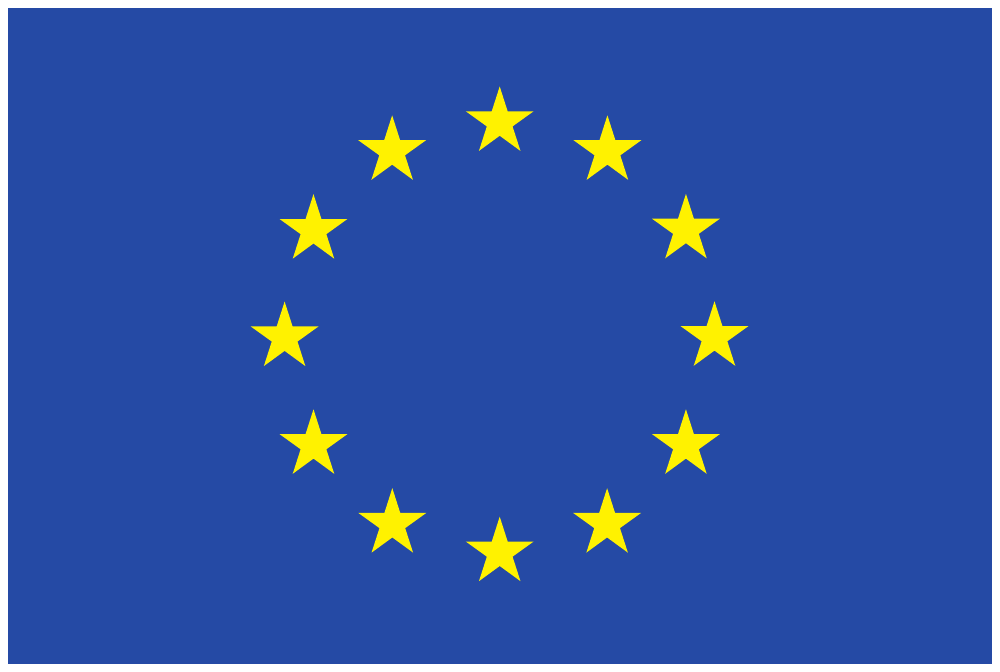}}}
\end{picture}

\section{Introduction}
\label{sec:introduction}

Vector addition systems with states (\vass or $d$-\vass when the dimension is $d$), also known as Petri nets, are a long established model for concurrent systems. 
The branching generalisation of VASS (\bvass or $d$-\bvass\footnote{The standard abbreviation for this model appearing in the literature is $\textsc{BVASS}$. 
Since in this work we work with both bounded and branching models of $\vass$, we prefer to disambiguate these variants by using prefixes $\textsc{Bo}$ and $\textsc{Br}$, respectively.}) 
is among the most popular extensions due to its multiple connections with database theory~\cite{BojanczykMSS09}, recursively parallel programs~\cite{BouajjaniE13}, timed pushdown systems~\cite{ClementeLLM17}, 
and many others.
A configuration of a $d$-\vass (or a $d$-\bvass) is a pair consisting of a state and a vector over $\N^d$. 
In $d$-\vass, the transitions update the configurations by changing the state and updating the configuration vector using vectors over $\Z^d$.
In $d$-\bvass, in addition there can be branching transitions, which create two independent copies of the system splitting the configuration vector among them.

The central decision problem for both \vass and \bvass is the \emph{reachability problem}, i.e., whether there exists a computation between given input configurations. 
For \vass we know that reachability is decidable~\cite{Mayr84}, moreover, recently the upper bound on the complexity was improved to Ackermann~\cite{LerouxS19}. 
Whether reachability is decidable for \bvass is one of the most intriguing open problems in formal verification, highlighted for example in the survey of Boja\'nczyk~\cite{Bojanczyk14a}.

To understand better the source of hardness, the reachability problem was studied when the dimension of the vectors is fixed to a constant. 
Then it matters how the vectors of the models are encoded and one must consider two cases: unary and binary encodings. 
Particularly, when the dimension is~$1$ it is known that $\onevass$ with vectors encoded in unary is NL-complete~\cite{ValiantP75}, while for the binary encoding it becomes NP-complete~\cite{HaaseKOW09}. 
For $\onebvass$ the complexities are also known, namely for the unary encoding the reachability  is PTIME-complete~\cite{GollerHLT16}, while for the binary encoding the problem is PSPACE-complete~\cite{FigueiraLLMS17}. 
We present these results in Table~\ref{table:complexity}. 
For these models the complexities are easy to remember as in both cases the complexity of the branching variant adds alternation to the complexity, as expected. 
However, branching is not alternation and this intuition is possibly misleading. 
Already in dimension~$2$ reachability for $\twovass$ is known to be NL-complete~\cite{EnglertLT16} and PSPACE-complete~\cite{BlondinFGHM15} for unary and binary encodings, respectively; 
while for $\twobvass$ it is not even known whether reachability is decidable.

\begin{table}
\centering
\caption{Complexities in dimension~$1$. In the last entry only a PSPACE lower bound is known.}\label{table:complexity}
\begin{tabular}{ c| c | c | c |}
& unary & binary & binary with a bound \\ \hline
$\onevass$ & NL-complete~\cite{ValiantP75} & NP-complete~\cite{HaaseKOW09} & PSPACE-complete~\cite{FearnleyJ15} \\  \hline
$\onebvass$ & PTIME-complete~\cite{GollerHLT16} & PSPACE-complete~\cite{FigueiraLLMS17} & in EXPTIME \\   \hline
\end{tabular}
\end{table}

In this paper we are interested in the variants of \vass and \bvass with vectors encoded in binary, where all values in all configurations are upper bounded by some natural number, given on input. 
We shall denote these classes \bovass and \bobr, respectively.
The model \bovass was popularised due to its connections with timed automata~\cite{HaaseOW12}, 
where it was shown that reachability for timed automata is interreducible in logarithmic space with reachability for \bovass. 
At the time it was known that for timed automata this problem is PSPACE-complete~\cite{AlurD94}, even when there are only 3 clocks available~\cite{CourcoubetisY92}, and NL-complete for the case with~1 clock; 
leaving open the complexity for timed automata with~2 clocks. 
The problem for two clocks was proved later to be PSPACE-complete~\cite{FearnleyJ15}, where the main technical contribution was proving that reachability for $\onebovass$ is PSPACE-complete. 
Then the final result for timed automata with two clocks followed from the reductions 
presented in~\cite{HaaseOW12}\footnote{Both papers~\cite{HaaseOW12,FearnleyJ15} write \emph{bounded counter automata} instead of \bovass. We write \bovass to have a uniform terminology for all models in this paper.}. 
Regarding the branching extension a more powerful version of the model $\onebobr$ was studied in~\cite{ClementeLLM17}, where connections between \bvass and pushdown timed automata were explored.

The reachability for $\onebovass$ has become one of the standard problems used in reductions proving PSPACE-hardness of various decision problems. 
Apart from the already mentioned application to prove PSPACE-hardness of reachability for two clock automata, $\onebovass$ was used e.g., to prove: PSPACE-completeness
of reachability for two dimensional \vass~\cite{BlondinFGHM15}; and PSPACE-completeness of regular separability of one counter automata~\cite{CzerwinskiL17}. 

\subparagraph*{Our contribution.}
In this paper we study the complexity of the reachability problem for \bobr. This problem can be easily proved to be in EXPTIME, but only PSPACE-hardness is known; e.g., 
due to the mentioned PSPACE-hardness of $\onebovass$. 
Our main contribution is that the reachability problem for $\twobobr$ is EXPTIME-complete and we leave the complexity of $\onebobr$ as an open problem (see Table~\ref{table:complexity}). 
We believe that if there exists an EXPTIME lower bound for reachability for $\onebobr$, then this would give a new, interesting starting problem for EXPTIME-hardness reductions. 
As an example application, we prove that reachability for $\onebobr$ is reducible in polynomial time to reachability in $\twobvass$, 
which slightly improves the state of art, as currently only a PSPACE lower bound is known for $\twobvass$~\cite{BlondinFGHM15}.

\subparagraph*{Organisation.} 
After presenting the main definitions and problems in Section~\ref{sec:preliminaries} and Section~\ref{sec:double}, we prove the EXPTIME-completeness of reachability for $\twobobr$ in Section~\ref{sec:hardness}. 
Then in Section~\ref{sec:limitations} we give some indications that using our techniques it would be hard to lift the EXPTIME-hardness result to $\onebobr$. 
We conclude in Section~\ref{sec:conclusion} providing a polynomial time reduction of reachability for $\onebobr$ to reachability for $\twobvass$.

\section{Preliminaries}
\label{sec:preliminaries}

\subparagraph*{VASS model.}
A \emph{Vector Addition System with States} ($\vass$) in dimension $d$ (in short $d$-$\textsc{VASS}$) is $\V = (Q,\Delta)$, 
where: $Q$ is a finite set of states and $\Delta \subseteq Q \times \Z^d \times Q$ is a finite set of transitions. The size of $\V$ is $|Q| + |\Delta| \cdot r$, where $r$ is the maximal representation size of vectors in $\Delta$ (written in binary).

A configuration is a pair $(q,\nn)$, denoted $q(\nn)$, where $\nn \in \N^d$.
Whenever $d=1$, we identify vectors in dimension~$1$ with natural numbers and in particular we denote configurations by $q(n)$ for $n \in \N$.
We will often refer to values of vector components as counter values and when convenient we will denote them $c_1,\ldots,c_d$.
We write that $q$ is the state of the configuration and $\nn$ is the vector of the configuration.
We write $p(\nn) \xrightarrow{t} q(\mm)$ if $t = (p, \zz, q) \in \Delta$ and $\nn + \zz = \mm$.
More generally, a {\em{run}} is a sequence of configurations 
\begin{align}\label{eq:run}
q_0(\nn_0) \xrightarrow{t_0} q_1(\nn_1) \xrightarrow{t_1} \ldots \xrightarrow{t_k} q_{k+1}(\nn_{k+1}),
\end{align}
where $t_i = (q_i,\zz_i,q_{i+1}) \in \Delta$ and $\nn_i + \zz_i = \nn_{i+1}$ for all $0 \le i \le k$. Notice that configurations have vectors over natural numbers and transitions have vectors over integers. We write $p(\nn) \xrightarrow{t_0\ldots t_k}q(\mm)$ or $p(\nn) \to^* q(\mm)$ if there exists a run like~\eqref{eq:run} such that $p(\nn) = q_0(\nn_0)$ and $q(\mm) = q_{k+1}(\nn_{k+1})$. We write $p(\nn) \not \to^* q(\mm)$ if no such run exists.

A bounded $d$-$\vass$ ($d$-$\bovass$) is a $d$-$\vass$ equipped with a bound $B$ such that numbers in all configurations cannot exceed $B$.
Formally, a $\bovass$ $\V$ is a tuple $(Q,\Delta,B)$ such that $(Q,\Delta)$ is a $\vass$ and $B \in \N$. We will assume that $B$ is given in binary. A run is defined in the same way like for $d$-$\vass$ but in~\eqref{eq:run} we additionally assume that $\nn_i \in \set{0,\ldots,B}^d$ for all $0\le i \le k$. That is, all components in all vectors are bounded by $B$.
Whenever speaking about bounded models of vector addition systems, we add the length of the bit representation of $B$ to the size of the system.

In the $d$-$\vass$ model one can simulate lower bound inequality tests as follows. We write that $t = (p,n,i,q) \in Q\times \N \times \set{1,\ldots,d} \times Q$ is an inequality test, and we allow to write $p(\nn) \xrightarrow{t} q(\mm)$ only if $\nn=\mm$ and $c_i\geq n$ in both configurations. We use the notation $t = (p,q)_{c_i \ge n}$ for simplicity.
In a $d$-$\vass$ $\V$ such a test can be simulated with one extra state $r$ and two transitions: $t_1 = (p,-\vv_i,r)$ and $t_2 = (r,\vv_i,q)$, where $\vv_i[j] = n$ for $i = j$ and $\vv_i[j] = 0$ otherwise. Then it is easy to see that $p(\nn) \xrightarrow{t} q(\mm)$ iff $p(\nn) \xrightarrow{t_1t_2} q(\mm)$ for all pairs of configurations $p(\nn)$ and $q(\mm)$.
If $\V$ is additionally a $\bovass$ with a bound $B \ge n$, then we can also implement analogously defined upper bound inequality tests $(p,q)_{c_i \le n}$. Indeed, it suffices to add one extra state $r$ and two transitions: $t_1 = (p,\ww_i,r)$ and $t_2 = (r,-\ww_i,q)$, where $\ww_i[j] = B-n$ for $i = j$ and $\ww_i[j] = 0$ otherwise.
Finally, equality tests $(p,q)_{c_i = n}$ can be encoded by a lower bound test followed by an upper bound test.
Slightly abusing the notation, we will assume that such inequality and equality tests are allowed in the transition set of $\bovass$ models, as they can be simulated as above at the cost of increasing the size of the system by a constant multiplicative factor.

\begin{definition}
\label{definition:compute}
Fix some function $f \colon \set{0,\ldots,M}^d \to \set{0,\ldots,N}^d$ for some $M,N \in \N$. We say that a $d$-$\bovass$ \emph{computes} $f$ if there exist states $p, q \in Q$ such that for all $\nn\in \set{0,\ldots,M}^d$ and $\mm\in \N^d$, we have $p(\nn) \to^* q(\mm)$ if and only if $\mm=f(\nn)$. The bound~$B$ of the $\bovass$ does not have to be equal to~$M$ or~$N$.
\end{definition}

Note that in the above definition, we do not specify how the $\bovass$ computing $f$ behaves starting in a configuration outside of the domain of $f$. 
However, it is easy to restrict the system to block on such configurations, \eg by adding tests $c_i \le M$ in the beginning.

\begin{example}
\label{example:copy}
Consider the copying function $f \colon \set{0,\ldots,M}^2 \to \set{0,\ldots,M}^2$ defined by $f(n,m) = (n,n)$ for all $0 \le n,m \le M$. We define a $2$-$\bovass$ $\V = (Q,\Delta,B)$ computing $f$ as follows. The bound is $B = M(M+2)$; the set of states is $Q = \set{p,r_1,r_2,q}$; the transitions are depicted in Figure~\ref{figure:copy}. To prove that this is the right function it suffices to observe that:
\begin{itemize}
 \item $p(n,m) \to^* r_1(k,l)$ iff $k = n$ and $l = 0$;
 \item $r_1(n,0) \to^* r_2(k,l)$ iff $k = 0$ and $l = (M+2)n$;
 \item and $r_2(0,(M+2)n) \to^* q(k,l)$ iff $k = l = n$.
\end{itemize}
\end{example}

\begin{figure}[!ht]
\centering
\begin{tikzpicture}[node distance = 2.8cm]
\node (p) {$p$};
\node[right of = p] (r1) {$r_1$};
\node[right of = r1] (r2) {$r_2$};
\node[right of = r2] (q) {$q$};

\path
(p) edge[->,loop above,>=stealth] node[above] {\footnotesize{$(0,-1)$}} (p)
(p) edge[->,>=stealth] node[above] {\footnotesize{$c_2 = 0$}} (r1)
(r1) edge[->,loop above,>=stealth] node[above] {\footnotesize{$(-1,M+2)$}} (r1)
(r1) edge[->,>=stealth] node[above] {\footnotesize{$c_1 = 0$}} (r2)
(r2) edge[->,loop above,>=stealth] node[above] {\footnotesize{$(1,-(M+1))$}} (r2)
(r2) edge[->,>=stealth] node[above] {\footnotesize{$c_2 \le M$}} (q)
;
\end{tikzpicture}
\caption{The $\bovass$ computing $f(n,m) = (n,n)$.}\label{figure:copy}
\end{figure}
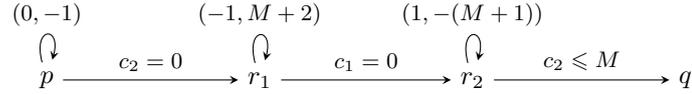

Note that in Example~\ref{example:copy}, the size of $\V$ is polynomial in the representation of $M$. A $2$-$\bovass$ of exponential size is trivial to define. In fact, for every function $f$ there exists a $\bovass$ computing $f$ of size exponential in the representation of the sizes of domain and codomain: one can simply use an equality test for every element of the domain to specify the corresponding value.


\subparagraph*{Branching VASS model.}
A \emph{Branching Vector Addition System with States} ($\bvass$) in dimension $d$ (in short $d$-$\textsc{BrVASS}$) is $\B = (Q,\Delta_1,\Delta_2,q_0)$, where: $Q$ is a finite set of states, $\Delta_1 \subseteq Q \times Z^d \times Q$ is the set of unary transitions, $\Delta_2 \subseteq Q^3$ is the set of binary transitions, and $q_0 \in Q$ is an initial state. The notions regarding configurations are the same as for $\vass$. A {\em{run}} of a $\bvass$ is a tree labelled with configurations such that internal nodes have either one or two children, and the following conditions hold.
\begin{itemize}
 \item For every internal node with one child, let $p(\nn)$ be its label and let $q(\mm)$ be the label of its child. Then there exists $(p,\zz,q) \in \Delta_1$ such that $\nn + \zz = \mm$;
 \item For every internal node with two children, let $p(\nn)$ be its label and let $q_1(\mm_1)$, $q_2(\mm_2)$ be the labels of its children. Then $(p,q_1,q_2) \in \Delta_2$ and $\nn = \mm_1 + \mm_2$.
  \item The leaves are labelled $q_0(\vzero)$, where $\vzero$ is the zero vector.
\end{itemize}
Notice that a $\bvass$ with $\Delta_2 = \emptyset$ is essentially a $\vass$ with a distinguished state $q_0$. A \emph{partial run} is a run that does not have to satisfy the last condition.

A {\em{bounded}} $d$-$\bvass$ ($d$-$\bobr$) is a $d$-$\bvass$ equipped with a bound $B\in \N$ that bounds all numbers in all configurations, similarly as for $d$-$\bovass$. As in the case of $\bovass$, for every $\bobr$ we abuse the notation allowing for inequality and equality test transitions in $\Delta_1$.

A $(p(\nn),q(\mm))$-context (or just context if the configurations are not relevant) of a $\bvass$ is a partial run such that the root of the tree is labelled $p(\nn)$ and one of the leaves is labelled $q(\mm)$; all other leaves are labelled $q_0(\vzero)$. We write $p(\nn) \to^* q(\mm)$ if there exists a $(p(\nn),q(\mm))$-context and we write $p(\nn) \not \to^* q(\mm)$ otherwise. Using this notation, the definition of computing a function $f \colon \set{0,\ldots,M}^d \to \set{0,\ldots,N}^d$ by a $d$-$\bobr$  generalises from the definition for $d$-$\bovass$ in the expected way.

%

For future reference, we state and prove the following easy claim.

\begin{lemma}\label{lem:increase-bound}
Given a $d$-$\bobr$ $\B$ with state set $Q$ and upper bound $B$ together with an integer $B'\geq B$, one can compute in polynomial time a $d$-$\bobr$ $\B'$ with state set $Q'\supseteq Q$ and upper bound $B'$ 
such that the following holds: for any $p,q\in Q$ and $n,m\in \set{0,\ldots,B}$, we have $p(n)\to^* q(m)$ in $\B$ if and only if $p(n)\to^* q(m)$ in $\B'$.
\end{lemma}
\begin{proof}
It suffices to modify $\B$ as follows: after each transition of $\B$, add an additional sequence of inequality tests checking that all counter values are bounded by $B$.
\end{proof}

\subparagraph*{The reachability problem.}
The reachability question for all mentioned models is, given a system within the model (say, a $\vass$ or a $\bobr$ etc.) and two configurations $p(\nn)$, $q(\mm)$, whether $p(\nn) \to^* q(\mm)$\footnote{For $\bvass$ usually this question is phrased asking for the existence of a run not a context. In this scenario it means that in the input of the problem one requires $q(\mm) = q_0(\vzero)$. It is easy to see that the two questions are equivalent in terms of complexity for all variants of $\bvass$ in this paper.}.

In this paper we are interested in the reachability problem for bounded models with vectors encoded in binary. 
For $d$-$\bovass$ (\ie without branching), the problem is easily contained in PSPACE for any $d$, because configurations can be described in polynomial space and solving the reachability problem amounts to nondeterministically guessing a run. 
Fearnley and Jurdzi\'nski proved in~\cite{FearnleyJ15} that, in fact, the problem is PSPACE-complete even for $d = 1$.

For the $\bobr$ model the situation is somewhat similar, however there is a gap in our understanding. 
On one hand, the problem is trivially in EXPTIME for any $d$, because the total number of configurations is exponential in the size of the input system. 
Note here, that even though again every configuration can be described in polynomial space, due to branching we cannot apply the same reasoning as for $\bovass$ to argue PSPACE membership.
For lower bounds, the result of~\cite{FearnleyJ15} shows PSPACE-hardness.
This leaves a gap between PSPACE and EXPTIME, which is the starting point of our work.

In this paper we show that the reachability problem for $d$-$\bobr$ model is EXPTIME-complete for any $d \geq 2$, leaving as an open problem the case of $d = 1$.

\begin{theorem}
\label{theorem:main}
For all $d\geq 2$, the reachability problem for $d$-$\bobr$ is EXPTIME-complete.
\end{theorem}

The proof of Theorem~\ref{theorem:main} is split into two parts. In Section~\ref{sec:double} we introduce auxiliary models that lie between $\onebobr$ and $\twobobr$. 
Later in Section~\ref{sec:hardness} we prove that these models are EXPTIME-hard.

\section{Extensions of $\onebobr$ with doubling and halving}
\label{sec:double}


We introduce convenient extensions of the models $\onebovass$ and $\onebobr$, where we add the ability of multiplying and dividing the counter value by $2$.
Formally, a {\em{doubling transition}} is a transition of the form $t = (p,q)_{\times 2}$ for some states $p,q$ that has the following semantics: we can write $p(n) \xrightarrow{t} q(m)$ if and only if $m = 2n$.
Similarly, a {\em{halving transition}} is a transition of the form $t = (p,q)_{\div 2}$ for some states $p,q$ that has the following semantics: we may write $p(n) \xrightarrow{t} q(m)$ if and only if $m = n / 2$. Note that in particular, if $n$ is odd then the transition cannot be applied.
Models $\doublebovass$, $\divbovass$, and $\doubledivbovass$ extend $\onebovass$ by respectively allowing doubling transitions, halving transitions, and both doubling and halving transitions.
Models $\divbobr$, $\doublebobr$ and $\doubledivbobr$ extend $\onebobr$ in the same~way.

It is not hard to see that doubling and halving transitions can be simulated using an additional counter, as explained next.

\begin{lemma}
\label{lem:plusone}
For every $\doubledivbovass$ $\V = (Q,\Delta)$ there exists a $\twobovass$ $\V' = (Q',\Delta')$ of polynomial size in $\V$ such that $Q \subseteq Q'$ and $p(n) \to^* q(m)$ iff $p(n,0) \to^* q(m,0)$, for all $p,q\in Q$.
Similarly, for every $\doubledivbobr$ there exists a $\twobobr$ that is analogous as above.
\end{lemma}
\begin{proof}
Every doubling transition $(p,q)_{\times 2}$ can be simulated using the extra counter, two additional states $r_1, r_2 \not \in Q$ and five transitions: $t_1 = (p, (0,0), r_1)$, $t_2 = (r_1,(-1,2),r_1)$, $t_3 = (r_1,r_2)_{c_1 = 0}$, $t_4 = (r_2,(1,-1),r_2)$, $t_5 = (r_2,q)_{c_2 = 0}$.
Similarly, every halving transition $(p,q)_{\times 2}$ can be simulated using the extra counter, two additional states $s_1, s_2 \not \in Q$ and five transitions: $t_1 = (p, (0,0), s_1)$, $t_2 = (s_1,(-2,1),s_1)$, $t_3 = (s_1,s_2)_{c_1 = 0}$, $t_4 = (s_2,(1,-1),s_2)$, $t_5 = (s_2,q)_{c_2 = 0}$.
\end{proof}

Therefore, $\doubledivbobr$ is a model in between $\onebobr$ and $\twobobr$. This model was inspired by the \emph{more powerful one register machine} (MP1RM)~\cite{schroeppel1972two}, used to prove the limitations of expressiveness for two-counter machines\footnote{The amazing name and abbreviation of the model come from~\cite{schroeppel1972two}.}. It is a one counter machine (without bounds) extended with two operations: multiplying by constants and dividing by constants. A more general model with polynomial updates was also studied in~\cite{FinkelGH13}.

By Lemma~\ref{lem:plusone}, the reachability problem for $d$-$\bobr$ for any $d\geq 2$ is at least as hard as the reachability problem for $\doubledivbobr$.
In the next section we will prove that the latter problem is already EXPTIME-hard, which will conclude the proof of Theorem~\ref{theorem:main}.
Before this, we show a side result which essentially says that doubling can be implemented using halving and vice versa.

\begin{proposition}
\label{proposition:double_is_division}
The reachability problems for models $\doubledivbovass$, $\doublebovass$ and $\divbovass$ are equivalent with respect to polynomial time reductions. 
Similarly, the reachability problems for models $\doubledivbobr$, $\doublebobr$ and $\divbobr$ are equivalent with respect to polynomial time reductions.
\end{proposition}
\begin{proof}
We will give polynomial-time reductions from model $\doubledivbovass$ to models $\doublebovass$ and $\divbovass$. The converse reductions are trivial, 
and the result for models with branching ($\doubledivbobr$, $\doublebobr$ and $\divbobr$) follows by applying the same construction.

We start with the reduction from $\doubledivbovass$ to $\doublebovass$.
Consider a $\doubledivbovass$ $\B$, say with an upper bound $B$.
By applying Lemma~\ref{lem:increase-bound}, we may assume that $B+1$ is a power of~$2$, say $B=2^n-1$ for some positive integer $n$.

The idea is to simulate every halving transition $(u,v)_{\div 2}$ in $\B$ with a sequence of $O(n)$ transitions (standard or doubling); since $n$ is polynomial in the size of $\B$, this is fine for a polynomial-time reduction. We explain this sequence in words; translating this description into a formal definition of a system and adding appropriate states and transitions to $\B$ is straightforward.
Apply $n-1$ times the following operation (depicted on Figure~\ref{figure:shiftleft}): if the counter value is smaller than $2^{n-1}$ then just double the counter; and otherwise subtract $2^{n-1}$ from the counter, double it and then add $1$.
Finally, at the end test whether the counter value is smaller than $2^{n-1}$.

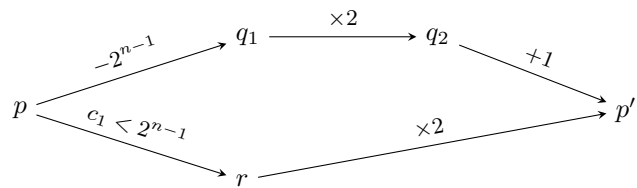
\begin{figure}[!ht]
\centering
\begin{tikzpicture}[node distance = 2.5cm]
\node (p) {$p$};
\node[above right = 0.5cm and 2.5cm  of p] (q1) {$q_1$};
\node[below right = 0.5cm and 2.5cm of p] (r) {$r$};
\node[right of = q1] (q2) {$q_2$};
\node[right = 7.5cm of p] (pp) {$p'$};

\path
(p) edge[->,>=stealth] node[sloped, above] {\footnotesize{$- 2^{n-1} $}} (q1)
(p) edge[->,>=stealth] node[sloped, above] {\footnotesize{$c_1 < 2^{n-1} $}} (r)
(q1) edge[->,>=stealth] node[sloped, above] {\footnotesize{$\times 2$}} (q2)
(q2) edge[->,>=stealth] node[sloped, above] {\footnotesize{$+1$}} (pp)
(r) edge[->,>=stealth] node[sloped, above] {\footnotesize{$\times 2$}} (pp)
;
\end{tikzpicture}
\caption{Shifting the bits cyclically to the left.}\label{figure:shiftleft}
\end{figure}

To see that these operations correctly implement halving the counter value, consider its bit encoding. 
Each single operation described above amounts to cyclically shifting the bit encoding to the left by $1$: every bit is moved to a position one higher, apart from the bit from the highest ($n$th) position which is moved to the lowest position. Thus, after $n-1$ applications the initial bit encoding is cyclically shifted by $n-1$ to the left, which is equivalent to shifting it by $1$ to the right. So now it suffices to verify that the highest bit (which was the lowest in the encoding of the initial counter value) is $0$, and if this is the case, then the final counter value is equal to half of the initial counter value.

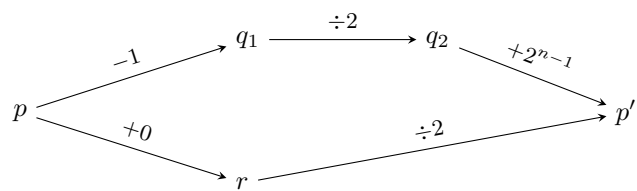
\begin{figure}[!ht]
\centering
\begin{tikzpicture}[node distance = 2.5cm]
\node (p) {$p$};
\node[above right = 0.5cm and 2.5cm  of p] (q1) {$q_1$};
\node[below right = 0.5cm and 2.5cm of p] (r) {$r$};
\node[right of = q1] (q2) {$q_2$};
\node[right = 7.5cm of p] (pp) {$p'$};

\path
(p) edge[->,>=stealth] node[sloped, above] {\footnotesize{$-1$}} (q1)
(p) edge[->,>=stealth] node[sloped, above] {\footnotesize{$+ 0$}} (r)
(q1) edge[->,>=stealth] node[sloped, above] {\footnotesize{$\div 2$}} (q2)
(q2) edge[->,>=stealth] node[sloped, above] {\footnotesize{$ + 2^{n-1}$}} (pp)
(r) edge[->,>=stealth] node[sloped, above] {\footnotesize{$\div 2$}} (pp)
;
\end{tikzpicture}
\caption{Shifting the bits cyclically to the right.}\label{figure:shiftright}
\end{figure}

The reduction from $\doubledivbovass$ to $\divbovass$ follows by the same reasoning, but reversed.
That is, we have to implement every doubling transition $(u,v)_{\times 2}$ in $\B$ using a sequence of $O(n)$ transitions (standard or halving).
First, we verify that the counter value is smaller than $2^{n-1}$, for otherwise the doubling transition cannot be applied.
If this is the case, we cyclically shift the binary encoding of the counter value by $1$ to the right $n-1$ times. As Figure~\ref{figure:shiftright} shows,
each such shift is performed by first nondeterministically guessing whether the counter value is odd, subtracting $1$ if this is the case, halving, and adding $2^{n-1}$ to the counter value if the counter value was odd.
\end{proof}

\section{Exptime-hardness}
\label{sec:hardness}

This section is devoted to prove EXPTIME-completness of $\doubledivbobr$. Note that Theorem~\ref{theorem:exptime_hardness} below together with Lemma~\ref{lem:plusone} immediately give Theorem~\ref{theorem:main}.

\begin{theorem}
\label{theorem:exptime_hardness}
The reachability problem for $\doubledivbobr$ is EXPTIME-complete.
\end{theorem}

We start by proving two lemmas that will be helpful in the reduction. Afterwards we define countdown games, an EXPTIME-complete problem from which we reduce. In the next two claims we write $M$ for some power of $2$, \ie $M = 2^n$ for some $n \in \N$.

\begin{lemma}
\label{lemma:xMx}
Let $f\colon \set{0,1,\ldots, M-1} \to \set{0,1,\ldots, M -1 + M(M-1)}$ be defined as $f(x) = x + Mx$. 
Then $f$ is computable by a $\doubledivbovass$ of size polynomial in the representation of $M$, with an upper bound $B = M + M(M-1) + M^2(M-1) + M^3(M-1)$.
\end{lemma}
\begin{proof}
The set of states is $Q = \set{p,r,q_0,q_1,\ldots,q_n}$. The transitions are depicted in Figure~\ref{figure:xMx}.

\begin{figure}[!ht]
\begin{tikzpicture}[node distance = 3cm]
\node (p) {$p$};
\node[right of = p] (r) {$r$};
\node[right of = r] (q1) {$q_0$};
\node[right of = q1] (dots) {$\ldots$};
\node[right of = dots] (qn) {$q_n$};

\path
(p) edge[->,loop above,>=stealth] node[above] {\footnotesize{$M+M^2+M^3$}} (p)
(p) edge[->,>=stealth] node[above] {\footnotesize{$0$}} (r)
(r) edge[->,loop above,>=stealth] node[above] {\footnotesize{$-1-M^3$}} (r)
(r) edge[->,>=stealth] node[above] {\footnotesize{$c_1 \le M^3 + M^2$}} (q1)
(q1) edge[->,>=stealth] node[above] {\footnotesize{$\div 2$}} (dots)
(dots) edge[->,>=stealth] node[above] {\footnotesize{$\div 2$}} (qn)
;
\end{tikzpicture}
\caption{The $\doubledivbovass$ computing $f(x) = x + Mx$.}\label{figure:xMx}
\end{figure}

We now prove that the system defined above indeed computes $f$. Consider any $x\in \set{0,1,\ldots,M-1}$.
Observe that
$p(x) \to^* r(y)$ if and only if $y=x + Ma + M^2a + M^3a - (M^3 + 1)b $ for some $0 \le b\le a \le M-1$ ($a$ is the number of loops taken in $p$ and $b$ is the number of loops taken in $r$).
Then observe that
$r(x + Ma + M^2a + M^3a) \to^* q_0(y)$ if and only if $y=x-a + Ma + M^2a$.
Finally, note that it is possible to go from $q_0$ to $q_n$ only if the initial counter value is divisible by $M$. 
Since $|x-a| < M$, this is possible only if $x = a$. Then we get $q_0(Mx + M^2x) \to^* q_n(y)$ if and only if $y=x + Mx$,
so in total $p(x)\to^* q_n(y)$ if and only if $y=x+Mx$, as required.
\end{proof}

We now combine the system provided by Lemma~\ref{lemma:xMx} with the branching feature of the $\bobr$ model to implement the key functionality: copying the counter value into two separate branches.

\begin{lemma}\label{lemma:branch}
There is a $\doubledivbobr$ of size polynomial in the representation of~$M$, with an upper bound as in Lemma~\ref{lemma:xMx}, and distinguished states $p,q_1,q_2$, with the following property.
For every $x\in \set{0,1,\ldots,M-1}$ there is a unique partial run $\rho_x$ of the system whose root is labelled with configuration $p(x)$ and whose all leaves have states $q_1$ or $q_2$.
Moreover, $\rho_x$ has exactly two leaves: one with configuration $q_1(x)$ and second with configuration $q_2(x)$.
\end{lemma}
\begin{proof}
We start the construction with the $\doubledivbovass$ provided by Lemma~\ref{lemma:xMx}, hence we can assume that there is a state $p'$ such that for all $x\in \set{0,1,\ldots,M-1}$, $p(x)\to^* p'(y)$ if and only if $y=x+Mx$. 
Next, add states $\set{r, q_1, s_0, \ldots, s_{n-1}, q_2}$ to the system together with the following transitions:
\begin{itemize}
\item branching transition $(p',r,s_0)$;
\item inequality test $(r,q_1)_{c_1\leq M-1}$; and
\item halving transitions $(s_i,s_{i+1})_{\div 2}$ for all $i\in \set{0,1,\ldots,n-1}$, where $s_n=q_2$.
\end{itemize}
We now argue that this system has the required property.
Observe that starting from configuration $p(x)$, the system first has to move to configuration $p'(x+Mx)$, where it branches into configurations $r(y)$ and $s_0(z)$ with $y+z=x+Mx$.
In one branch we test that $y\leq M-1$ and end in configuration $q_1(y)$. In the other branch we divide $z$ by $2$ exactly $n$ times, which amounts to testing that $z$ is divisible by $M$ and finishing in configuration $q_2(z/M)$. Since $y\leq M-1$, $z$ is divisible by $M$, and $x\in \set{0,1,\ldots,M-1}$, we conclude that the split $(y,z)$ of $x+Mx$ has to be equal to $(x,Mx)$.
Hence, the only two leaves are labelled with configurations $q_1(x)$ and $q_2(x)$, respectively.
\end{proof}

We will now show a reduction from a classic EXPTIME-complete problem --- determining a winner of a countdown game --- to reachability in $\doubledivbobr$. A {\em{countdown game}} is a pair $(S,T)$ such that $S$ is a set of nodes partitioned into two subsets $S = S_\exists \uplus S_\forall$, and $T \subseteq S \times (\N \setminus \set{0}) \times S$ are weighted transitions, with weights encoded in binary.
We can assume that for every node $s$ there exists exactly two outgoing transitions $(s,n,s')$ for some integer $n$ and node $s'$.
A configuration of a game is $s(c)$, where $s \in S$ and $c \in \N$.

The game has two players, each owning her set of nodes (Player~$q$ owns $S_q$, for $q\in \{\exists,\forall\}$). 
The game proceeds as follows: for every configuration $s(c)$ the owner of the node $s$ chooses one of the two transitions $(s,n,s')$ such that $n \le c$ and moves to the configuration $s'(c - n)$.
If it is not possible to choose such a transition, then the game ends and Player~$\forall$ wins the game.
Otherwise, if at some point a configuration with value~$0$ is reached, then the game ends and Player~$\exists$ wins the game.

Given a countdown game $(S,T)$ and a starting configuration $s(c)$ the problem whether Player~$\exists$ has a winning strategy is known to be EXPTIME-complete~\cite{JurdzinskiSL08}\footnote{The game considered in~\cite{JurdzinskiSL08} is slightly different but it is easy to show that both definitions are equivalent.}. Without loss of generality we can assume that $c = 2^n-1$ for some integer $n$.

\begin{proof}[Proof of Theorem~\ref{theorem:exptime_hardness}]
We give a reduction from the problem of determining the winner of a countdown game.
Let then $(S,T)$ be the given countdown game and $s(c)$ be its starting configuration, where $c=2^n-1$ for some integer $n$.
We are going to construct, in polynomial time, a $\doubledivbobr$ $\B$ together with a root configuration such that $\B$ has an accepting run from the root configuration (\ie one where all leaves are labelled with the initial configuration $q_0(0)$) if and only if Player~$\exists$ wins the game.

We start by including all nodes of $S$ in the state set of $\B$.
For every node $r\in S_\exists$, let $(r,n_1,r_1')$ and $(r,n_2,r_2')$ be the two transitions in $T$ that have $r$ as the origin.
Then add transitions $(r,-n_1,r_1')$ and $(r,-n_2,r_2')$ to the transition set of $\B$.

Next, for every node $r\in S_\forall$, let $(r,n_1,r_1')$ and $(r,n_2,r_2')$ be the two transitions in $T$ that have $r$ as the origin.
Then add a copy of the system provided by Lemma~\ref{lemma:branch} for $M=c+1$, and identify its distinguished state $p$ with $r$.
Denoting the other two distinguished states of the copy by $q_1^r$ and $q_2^r$, add transitions $(q_1^r,-n_1,r_1')$ and $(q_2^r,-n_2,r_2')$ to $\B$.

Finally, we add an initial state $q_0$ together with a transition $(r,0,q_0)$ for each $r\in S$.
Thus, $\B$ may accept whenever it reaches any state $r\in S$ with counter value $0$.
This finishes the construction of $\B$. We set the root configuration to be $s(c)$.

To see that the reduction is correct it suffices to observe that accepting runs of $\B$ with root configuration $s(c)$ are in one-to-one correspondence with winning strategies of Player~$\exists$ in the game $(S,T)$ starting from $s(c)$.
Indeed, in nodes belonging to Player~$\exists$, she may choose any of the two transitions originating there, which corresponds to the possibility of taking any of the two transitions in an accepting run of $\B$.
In nodes belonging to Player~$\forall$, Player~$\exists$ has to be able to win for both possible moves of Player~$\forall$, which corresponds to requiring acceptance in both subtrees obtained as a result of running the respective copy of the system provided by Lemma~\ref{lemma:branch}.
\end{proof}

\section{Limitations of $\onebobr$}
\label{sec:limitations}

The natural approach to show EXPTIME hardness for $\onebobr$ is to ask whether they can efficiently compute the function $f : \set{0,\ldots,M} \to \set{0,\ldots, 2M}$ defined by $f(x) = 2x$. 
Then $\onebobr$ could simulate $\doublebobr$ and EXPTIME-hardness would follow from Theorem~\ref{theorem:exptime_hardness} and Proposition~\ref{proposition:double_is_division}.
Unfortunately, we show that this is not the case. In the following proposition, we assume the original definition of $\onebobr$, without inequality or equality tests as single transitions.

\begin{proposition}
\label{proposition:no_double}
Consider the function $f \colon \set{0,\ldots,M} \to \set{0,\ldots, 2M}$ defined as $f(x) = 2x$, for some $M \in \N$. 
Then every $\onebobr$ computing the function $f$ has at least $M$ states.
\end{proposition}
\begin{proof}
Suppose $\B$ is a $\onebobr$ that computes $f$: there are $p,q\in Q$ such that for all $x\in \set{0,\ldots,M}$, $p(x)\to^* q(y)$ if and only if $y=2x$.
Here, $Q$ denotes the set of states of $\B$.

For every $x\in \set{0,\ldots,M}$, fix any run $\rho_x$ that witnesses $p(x)\to^* q(2x)$. The {\em{spine}} of $\rho_x$ is the path from the root (labelled $p(x)$) to the leaf labelled $q(2x)$ (in the special case when $q(2x) = q_0(0)$ this path might be not unique and the spine is defined as any path).

\begin{claim}
For every $x\in \set{1,\ldots,M}$, the spine of $\rho_x$ contains at least one configuration with counter value $0$.
\end{claim}
\begin{proof}
Suppose every configuration on the spine of $\rho_x$ has a positive counter value.
Consider run $\rho'$ obtained from $\rho_x$ by decrementing all the counter value in each configuration on the spine by $1$.
Then $\rho'$ witnesses that $p(x-1)\to^* q(2x-1)$, contradicting the assumption that $\B$ computes $f$.
\cqed\end{proof}

\begin{claim}
For all different $x,x'\in \set{0,\ldots,M}$, the sets of configurations on the spines of $\rho_x$ and $\rho_{x'}$ are disjoint.
\end{claim}
\begin{proof}
Suppose some configuration $r(z)$ appears both on the spine of $\rho_x$ and of $\rho_{x'}$.
Then replacing the context rooted at $r(z)$ in $\rho_x$ with the context rooted at $r(z)$ in $\rho_{x'}$ yields a run witnessing that $p(x)\to^* q(2x')$.
As $x\neq x'$, this contradicts the assumption that $\B$ computes~$f$.
\cqed\end{proof}

The proposition now follows immediately from the two claims above: For every $x\in \set{1,\ldots,M}$ the spine of $\rho_x$ has to contain some configuration with counter value $0$,
and these configuration have to be pairwise distinct for different $x$. This implies that $\B$ has to have at least $M$ states, as claimed.
\end{proof}

\section{Conclusion}
\label{sec:conclusion}

We have proved EXPTIME-completeness of reachability for $\twobobr$ leaving the complexity for $\onebobr$ between PSPACE and EXPTIME.
Our EXPTIME-hardness proof for $\twobobr$ heavily relied on the ability of manipulating a counter by multiplying and dividing it by $2$.
As shown by Proposition~\ref{proposition:no_double}, one counter alone is not enough to implement this functionality in the $\bobr$ model.

At this moment we do not have any clear indication on whether reachability for $\onebobr$ is EXPTIME-hard, or it rather leans towards PSPACE.
On one hand, EXPTIME-hardness would give us a natural problem to reduce from for other infinite-state systems with branching, similarly to the role served by reachability for $\onebovass$.
On the other hand, PSPACE-membership would be even more interesting, as intuitively it would show that in the $\onebobr$ model the branching transitions are too weak to emulate alternation,
contrary to the situation in multiple other models of similar kind.

To motivate the study of this problem even further, we highlight the connections between $\onebobr$ and $\twobvass$.
It is known that reachability for $\twovass$ is PSPACE-hard, however the only proof of hardness we are aware of uses the reduction from $\onebovass$. 
We show that similarly one can reduce $\onebobr$ to $\twobvass$.

\begin{proposition}
\label{proposition:2bvass}
The reachability problem for $\onebobr$ is reducible to the reachability problem for $\twobvass$ in polynomial time.
\end{proposition}
\begin{proof}
Consider a $\onebobr$ $\B$, say with state set $Q$ and upper bound $B\in \N$.
The idea is to construct a $\twobvass$ $\V$ where every configuration $p(x)$ of $\B$ will be simulated by a configuration of the form $p(x,2B-x)$ of $\V$.

We start by including $Q$ in the state set of $\V$. 

Next, we model every unary transition $t=(u,k,v)$ of $\B$ using two unary transitions in~$\V$: 
$(u,(k,-B-k),r_t)$ followed by $(r_t,(0,B),v)$, where $r_t$ is a new state added for the transition~$t$. 
Note that this sequence of two transitions can be fired in configuration of the form $u(x,2B-x)$ if and only if $x+k\in \set{0,\ldots,B}$, 
and then the resulting configuration is $v(x+k,2B-(x+k))$.

Every branching transition $t=(u,v_1,v_2)$ of $\B$ is modelled in~$\V$ as follows. We add to~$\V$ two new states $s_{t,1}$ and $s_{t,2}$, a branching transition $(u,s_{t,1},s_{t,2})$, and unary transitions
$(s_{t,1},(0,B),v_1)$ and $(s_{t,2},(0,B),v_2)$. Observe that thus, after firing transition $(u,s_{t,1},s_{t,2})$ in configuration of the form $u(x,2B-x)$ we obtain two configurations of the form 
$s_{t,1}(x_1,y'_1)$ and $s_{t,2}(x_2,y'_2)$ satisfying $x_1+x_2=x$ and $y'_1+y'_2=2B-x$, in which we have to apply the transitions that add $B$ to the second counter, 
obtaining configurations $v_1(x_1,y_1)$ and $v_2(x_2,y_2)$ satisfying $x_1+x_2=x$ and $y_1+y_2=4B-x$. Hence, provided the initial split of the second counter was $(y'_1,y'_2)=(B-x_1,B-x_2)$, 
we indeed obtain configurations $v_1(x_1,2B-x_1)$ and $v_2(x_2,2B-x_2)$, as was our initial intention. We will later argue that in runs in which we are interested, 
in all branching transitions the second counter has to be split as above.

Finally, if $q_0$ is the initial state of $\B$, we add a new initial state $q_0'$ to $\V$ together with a transition $(q_0, (0,-2B), q_0')$. 
Thus, the only way to arrive at configuration $q_0'(0,0)$ is to apply this transition in configuration $q_0(0,2B)$.

Now consider any pair of configurations $p(n),q(m)$ of $\B$, where $n,m\in \set{0,\ldots,B}$. We claim that $p(n)\to^* q(m)$ in $\B$ 
if and only if $p(n,2B-n)\to^* q(m,2B-m)$ in $\V$.

For the forward direction, consider a partial run $\pi$ of $\B$ witnessing $p(n)\to^* q(m)$ and modify it to a run $\rho$ of $\V$ as follows. 
First, replace every configuration of the form $u(x)$ in $\pi$ with configuration $u(x,2B-x)$.
Next, replace the transitions of $\B$ in $\pi$ with sequences of transitions of $\V$ in the expected way, as described above.
In particular, whenever in $\pi$ a branching transition $(u,v_1,v_2)$ is applied to configuration $u(x)$ and results in configurations $v_1(x_1)$ and $v_2(x_2)$,
in $\rho$ we model this by applying a branching transition to $u(x,2B-x)$ and obtaining $s_{t,1}(x_1,B-x_1)$ and $s_{t,2}(x_2,B-x_2)$, which are subsequently modified to $v_1(x_1,2B-x_1)$ and $v_2(x_2,2B-x_2)$.
Then $\rho$ witnesses that $p(n,2B-n)\to^* q(m,2B-m)$ in $\V$.

For the backward direction, consider a partial run $\rho$ of $\V$ witnessing $p(n,2B-n)\to^* q(m,2B-m)$.
A node of $\rho$ will be called {\em{principal}} if the state of its configuration belongs to the original state set $Q$; other nodes will be called {\em{auxiliary}}.
Note that in $\rho$, every second level consists entirely of principal nodes, and every other level consists entirely of auxiliary nodes.
More specifically, every principal node $a$, say labelled with a configuration $u(x,y)$, is of one of the following four kinds:
\begin{itemize}
\item $a$ has one child and one grandchild, labelled respectively with configurations $r_t(x+k,y-k-B)$ and $v(x+k,y-k)$ for some unary transition $(u,k,v)$ of $\B$;
\item $a$ has two children, each having one child of its own, and these are respectively labelled with configurations $s_{t,1}(x_1,y'_1),s_{t,2}(x_2,y'_2),v_{1}(x_1,y'_1+B),v_{2}(x_2,y'_2+B)$ 
for some branching transition $(u,v_1,v_2)$ of $\B$;
\item $a$ has one leaf child labelled with $q_0'(0,0)$, so in particular the configuration at $a$ is $q(x,z)=q_0(0,2B)$; or
\item $a$ is the unique leaf labelled with $q(m,2B-m)$.
\end{itemize}
These types of principal nodes will be called {\em{unary}}, {\em{branching}}, {\em{leaf}}, and {\em{final}}, respectively.

We first prove by a bottom-up induction over $\rho$ that in all principal nodes, the sum of counter values is equal to $2B$.
This trivially holds for all principal nodes that are leaf or final.
Next, whenever $a$ is a unary principal node with configuration $u(x,y)$, then by the induction assumption for its grandchild we infer that $(x+k)+(y-k)=2B$ for some integer~$k$, implying $x+y=2B$ as required.
The case of branching principal nodes is analogous.

We now prove by a top-down induction over $\rho$ that in all principal nodes, the value of the first counter is upper bounded by $B$.
This is certainly true in the root, since we assumed that $n\leq B$.
Suppose $a$ is a unary principal node with configuration $u(x,2B-x)$ for which we already know that $x\leq B$, 
and its child and grandchild are respectively labelled with $r_t(x+k,B-x-k)$ and $v(x+k,2B-x-k)$ for some unary transition $(u,k,v)$ of $\B$.
Then the existence of the child's configuration implies that $B-x-k\geq 0$, implying that $x+k\leq B$, as required.
Finally, observe that if $a$ is a branching principal node, then the sum of the first counter values in its grandchildren is equal to the first counter value in $a$, hence an upper bound of $B$ on the latter
entails an upper bound of $B$ on the former.

All in all, we argued that all principal nodes in $\rho$ have configurations of the form $u(x,2B-x)$, where $u\in Q$ and $x\in \set{0,\ldots B}$.
It now remains to observe that taking all principal nodes with tree order induced from $\rho$ and forgetting the second counter value yields a run of $\B$ witnessing that $p(n)\to^* q(m)$ in $\B$.
\end{proof}

\subparagraph*{Acknowledgements.}
The authors would like to thank Marthe Bonamy for feeding them during the work on this project.
No beavers were harmed in the making of this paper.



\bibliography{bobrvas}

\end{document}